\documentclass[journal,technote,twocolumn]{IEEEtran}

\usepackage{amssymb,amsmath,amsthm,amsfonts}
\usepackage{bm}
\usepackage{amsmath}

\newtheorem{theorem}{Theorem}

\newtheorem{definition}{Definition}
\newtheorem{example}{Example}
\newtheorem{lemma}{Lemma}
\newtheorem{proposition}{Proposition}
\newtheorem{corollary}{Corollary}
\newtheorem{remark}{Remark}

\newcommand{\mm}{\mathrm}
\newcommand{\ud}{\mm{d}}

\def\span{\mathop{\sf span}}

\def\T{\intercal}

\usepackage{xcolor}

\usepackage[numbers,sort&compress,square]{natbib}

   \usepackage[pdftex]{graphicx}
   \DeclareGraphicsExtensions{.pdf,.jpeg,.png}

\begin{document}
\title{
Structural accessibility and structural  observability of nonlinear networked systems
}

\author{
Marco~Tulio~Angulo$^{*}$, 
Andrea~Aparicio and
         Claude H.~Moog  
\thanks{M.T. Angulo is with CONACyT - Institute of Mathematics,  Universidad Nacional Autonoma de Mexico, Juriquilla Mexico. Correspondence should be addressed to mangulo@im.unam.mx.}
\thanks{A. Aparicio is with the Institute of Mathematics,  Universidad Nacional Autonoma de Mexico, Juriquilla Mexico.}%
\thanks{C. Moog is with the Laboratoire des Sciences du Num\'erique de Nantes, UMR CNRS 6004, Nantes 44321, France.}
}

\maketitle

\begin{abstract}
The classical notions of structural controllability and structural observability are receiving increasing attention in Network Science, since they provide a mathematical basis to answer how the network structure of a dynamic system affects its controllability and observability properties. 
However, these two notions are formulated assuming systems with linear dynamics, which significantly limit their applicability. 
To overcome this limitation, here we introduce and fully characterize the notions  ``structural accessibility'' and ``structural observability'' for systems with nonlinear dynamics. 
%
We show how nonlinearities make easier the problem of controlling and observing networked systems, reducing the number of variables that are necessary to directly control and directly measure. Our results contribute to understanding better the role that the network structure and nonlinearities play in our ability to control and observe complex dynamic systems.
\end{abstract}

\begin{IEEEkeywords}
networks; nonlinear systems; observability; accessibility; controllability.
\end{IEEEkeywords}

\IEEEpeerreviewmaketitle


\section{Introduction}

In a world where complex networks underlie most biological, social and technological systems that shape the human experience \cite{Barabasi:15,motter2015networkcontrology}, one central challenge is finding principles that can help us control and observe complex networked systems.  
When only the network structure of a dynamical system is known (i.e., a graph of the interactions between its variables), a central theoretical basis for this research program has been the classical notions of ``structural controllability'' and ``structural observability'' of linear systems \cite{liu2016control}. 
These two notions characterize the conditions under which almost all linear dynamical systems whose structure matches a given network are controllable or observable, respectively \cite{dion2003generic}. 
Linear structural controllability and structural observability thus provide a mathematical formalism for predicting how changes in the network structure of a system impact its controllability and observability properties. 
For example,  linear structural controllability was applied to build and then experimentally validate predictions of how removing different neurons (i.e., removing nodes in the network) affects the locomotion of the  round worm  \emph{C. elegans}  \cite{yan2017network}.  
Additionally, over the last few years, a central line of research has been characterizing minimal sets of ``driver nodes'' and ``sensor nodes'' from which we can efficiently render a complex networked system controllable and observable \cite{liu2016control}.

The conditions of  linear structural controllability and linear structural observability can be stronger than necessary when applied to systems with nonlinear dynamics, resulting in over-conservative predictions.
This is because  the lack of linear controllability (resp. linear observability) of a nonlinear system cannot be used to predict its lack of controllability (resp. observability).
An elementary example of this is a car, which is controllable but not linearly controllable because it cannot  move in the direction of the axis defined by its rear wheels.
Despite the ubiquity of nonlinear systems in nature and technology, the effects of nonlinearities on our ability to efficiently control and observe complex networked systems remain poorly understood \cite{liu2013observability,whalen2015observability,gates2016control,zanudo2017structure,haber2017state}.

%
%


Given that most systems in nature are expected to contain nonlinearities, in this Note we introduce and characterize the notions of nonlinear ``structural accessibility'' and nonlinear ``structural  observability'' as  counterparts of linear structural controllability and linear structural observability. 
These two notions we introduce characterize the conditions under which almost all nonlinear systems whose structure matches a given network are locally accessible or  locally observable almost everywhere, respectively. 
Accessibility and observability are nonlinear generalizations of  linear controllability and linear observability, which have played a central role in the development of nonlinear control theory \cite{conte2007algebraic}.
%
Somewhat counter-intuitively, 
we show that nonlinearities make  significantly easier the problem of controlling or observing a complex networked system. 
More precisely, our main result proves that the conditions for nonlinear structural accessibility and  observability are  weaker than the conditions for linear structural controllability and observability. 
%
%
%
We show this implies that we need smaller sets of driver and sensor nodes  when compared to the those necessary for  linear structural controllability and linear structural observability.

This Note is organized as follows. 
Section II summarizes the network characterization of structural controllability and structural observability of linear systems, serving as a comparison point to our results.
Section III contains our problem statement and main results. Proofs are collected in Sections IV and V. We end discussing some predictions that our structural accessibility theory offers about the locomotion of \emph{C. elegans}, and some limitations of our approach.


\section{Preliminaries}
\label{statement-mainresults}

%
The network or \emph{graph} of a system  with $N$ state variables, $M$ inputs, and $P$ outputs is a directed graph 
$\mathcal G =  ({\sf X \cup \sf Y \cup \sf U}, {\sf A \cup B \cup C})$ containing state nodes ${\sf X} = \{{\sf x}_1 , \cdots, {\sf x}_N\}$, output nodes ${\sf Y} = \{{\sf y}_1, \cdots, {\sf y}_P\} $, and  input nodes ${\sf U} = \{ {\sf u}_1, \cdots, {\sf u}_M\} $, see Fig. 1a.
%
%
Edges  take the form $({\sf x}_j \rightarrow {\sf x}_i) \in {\sf A}$ to denote that the $i$-th state variable directly depends on the $j$-th one, $({\sf x}_j \rightarrow {\sf y}_i) \in {\sf C}$ to denote that the $i$-th measured output directly depends on the $j$-th state variable, and    $({\sf u}_j \rightarrow {\sf x}_i) \in {\sf B}$ to denote  that the $i$-th state variable directly depends on the $j$-th control input.
%
%
We  allow graphs with empty output or input node sets to represent systems without  outputs or inputs, respectively.

In the framework of linear structural controllability and linear structural observability  the system dynamics is of course assumed linear. Then   the controllability and observability  of the set of all linear systems whose structure matches the graph $\mathcal G$ is analyzed. 
More precisely, the system dynamics is assumed to have the form
\begin{equation}
\label{eq:linear}
 \dot x(t) = Ax(t) + Bu(t), \quad y(t) = Cx(t),
\end{equation}
where $x(t) \in \mathbb R^N$, $u(t) \in \mathbb R^M$ and $y(t) \in \mathbb R^P$ are the state, input, and output of the system at time $t$, respectively. Here  $A = (a_{ij}) \in \mathbb R^{N \times N}$, $B = (b_{ij}) \in \mathbb R^{N \times M}$ and $C = (c_{ij}) \in \mathbb R^{P \times N}$ are matrices of parameters.
The \emph{structure} of Eq. \eqref{eq:linear} is determined by the zero/non-zero pattern of these three matrices.
%
%
Thus, given a graph $\mathcal G$,  the class $\mathfrak D_L(\mathcal G)$  of all linear systems whose structure matches $\mathcal G$ is defined as all systems \eqref{eq:linear} such that: $a_{ij} \neq 0$ iff $({\sf x}_j \rightarrow {\sf x}_i) \in {\sf A}$, 
$b_{ij} \neq 0$  iff  $({\sf u}_j \rightarrow {\sf x}_i) \in {\sf B}$, and $c_{ij} \neq 0$ iff $({\sf x}_j \rightarrow {\sf y}_i) \in {\sf C}$.
%
Note that the edges $({\sf x}_j \rightarrow {\sf x}_i)$ and $({\sf u}_j \rightarrow {\sf x}_i)$ are encoded by  differential equations. By contrast, the edges $({\sf x}_j \rightarrow {\sf y}_i)$ are encoded by  algebraic equations; these output edges have direction because the output map $y = Cx$ is not necessarily one-to-one (e.g., the single output $y = x_1 + x_2$).
Thus, the class  $\mathfrak D_L(\mathcal G)$  describes the set of all linear dynamics that a system can take if its structure coincides with $\mathcal G$.

The class $\mathfrak D_L(\mathcal G)$ is said \emph{structurally controllable} (resp. \emph{structurally observable}) if 
 it contains at least one system that is linearly controllable (resp. linearly observable) \cite{dion2003generic}.
In that case we also say that $\mathcal G$ is linearly structurally controllable (resp. linearly structurally observable).
It turns out that when one system in $\mathfrak D_L(\mathcal G)$ is linearly controllable (resp. linearly observable), then almost all other systems in $\mathfrak D_L(\mathcal G)$ are linearly controllable as well (resp. linearly observable) \cite{dion2003generic}.
This means that, if $\mathfrak D_L(\mathcal G)$ is structurally controllable (resp. structurally observable), any of its systems  is either controllable (resp. observable), or becomes controllable (resp. observable) by an infinitesimal change in the nonzero entries of the matrices $A,B$ and $C$.
%
%
%
A central result in the theory of structural linear systems, which can be traced back to the pioneer work of Lin in the 70's \cite{lin1974structural},  is the following:
\begin{theorem}(see, e.g., \cite{dion2003generic}) $\mathfrak D_L(\mathcal G)$ is:
\begin{itemize}
\item[(i)]  structurally controllable iff each state node  is the end-node of a path that starts in ${\sf U}$; and there is a disjoint union of cycles and paths starting in $\sf U$ that covers  ${\sf X}$.
\item[(ii)] structurally observable iff each state node  is the start-node of a path that ends in ${\sf Y}$; and there is a disjoint union of cycles and paths ending in $\sf Y$ that covers  ${\sf X}$.
\end{itemize}
\end{theorem}

Recall that a \emph{path} is a sequence of nodes  ${\sf v}_1 \rightarrow {\sf v}_2 \rightarrow \cdots \rightarrow {\sf v}_n$ where ${\sf v}_i \in \sf X \cup Y \cup U$.
The \emph{start-node} of this path is ${\sf v}_1$ and its \emph{end-node}  is ${\sf v}_n$. 
A \emph{cycle} is a path that starts and ends in the same node (i.e., ${\sf v}_n = {\sf v}_1$).
Two paths are \emph{disjoint} if they have  disjoint sets of nodes.
%

Theorem 1 shows that except for a zero-measure set of ``singularities,'' the graph $\mathcal G$ of a linear system determines its controllability and observability properties. 
Note that for linear structural controllability it is not sufficient that the control inputs propagate their influence through $\mathcal G$ to all state nodes. 
Similarly, for linear structural observability, it is not sufficient that each state node can propagate its state to some output through $\mathcal G$. 
Both notions require that the graph $\mathcal G$ contains enough ``independent'' paths to propagate these effects, encoded by the existence of a disjoint union of cycles and paths that covers all state nodes.

\begin{example} 
\label{ex:example-1}
For the graph $\mathcal G$ of Fig. 1a, the class $\mathfrak D_L(\mathcal G)$ contains  all linear systems of the form
\begin{equation}
\label{eq:linear-dilation}
\left \{
\begin{split}
\dot x_1(t) &= b_{11} u_1(t), \\ 
\dot x_2(t) &= b_{21} u_1(t), \\
 y_1(t) &= c_{11} x_1(t) + c_{12} x_2(t),
\end{split}
\right.
\end{equation}
with nonzero constants $b_{11}, b_{21}, c_{11}$ and $c_{12}$. 
Recall that:
\begin{itemize}
\item [1.] Together with isolated nodes in $\mathcal G$,  the main obstacle for  linear structural controllability is the presence of so-called ``dilations'' \cite{lin1974structural}. In  essence, a dilation consist of two nodes with identical dynamics that are controlled by the same input (top in Fig. 1a).
A dilation makes $\mathcal G$ not structurally controllable because it is impossible to obtain a disjoint union of paths that covers  $\sf X$. 
For Fig. 1a, all systems in $\mathfrak D_L(\mathcal G)$ are uncontrollable because their state is constrained to the plane $b_{11} x_2(t) - b_{21} x_1(t) = b_{11} x_2(0) - b_{21} x_1(0)$ for all inputs $u_1(t)$ and time $t$ (Fig. 1b). 
%

\item[2.] Analogously, so-called ``contractions'' in $\mathcal G$ are the main obstacle for linear structural observability. In  essence, a contraction corresponds to two state nodes that are measured using a single output (bottom in Fig. 1a).
Indeed, for Fig. 1a, all systems  $\mathfrak D_L(\mathcal G)$ are unobservable because  using $y_1 = c_{11} x_1 + c_{12} x_2$ and $k$ of its derivatives $y_1^{(k)} = (c_{11} b_{11} + c_{12} b_{12}) u_1^{(k)}$ it is impossible to infer the value of $x_1$ and $x_2$ (Fig. 1c). 
%
%
%
\end{itemize}

\end{example}

Theorem 1 provides a theoretical  basis for a very active research line aiming to identify and analyze the ``driver'' and ``sensor'' nodes  that render a system linearly structurally controllable and linearly structurally observable  (see, e.g., \cite{motter2015networkcontrology,liu2016control}). 
%
%
%
%
More precisely, consider a graph $\mathcal G({\sf X}, {\sf A})$ with only  state nodes $\sf X$ and edges $({\sf x}_i \rightarrow {\sf x}_j) \in  {\sf A}$.
Then define:
\begin{definition}
\label{def:driver-sensor-nodes}
 \hfill
\begin{itemize}
 \item[(i)]${\sf X}_D \subseteq \sf X$ is a set of \emph{driver nodes}  if  there exists a set $\sf U$ of input nodes and a set $\sf B$ of edges  of the form $({\sf u}_i \rightarrow {\sf x}_j)$ such that: (i) the graph $\mathcal G({\sf X \cup \sf U}, \sf A \cup B)$ is linearly structurally controllable; and (ii) all and only the driver nodes have incoming edges from the input nodes (i.e.,  $({\sf u}_i \rightarrow {\sf x}_j) \in  {\sf B}$ iff ${\sf x}_j \in {\sf X}_D$).
 \item[(ii) ] ${\sf X}_S \subseteq \sf X$ is a set of \emph{sensor nodes}  if  there exists a set $\sf Y$ of output nodes and a set  $\sf C$ of edges  of the form $({\sf x}_i \rightarrow {\sf y}_j)$ such that: (i) the graph $\mathcal G({\sf X \cup \sf Y}, \sf A \cup C)$ is linearly structurally observable; and (ii) all and only the sensor nodes have outgoing edges to the output nodes (i.e., $({\sf x}_i \rightarrow {\sf y}_j) \in  \sf C$ iff ${\sf x}_i \in {\sf X}_S$).
 \end{itemize}
 \end{definition}
 
A set of driver nodes or sensor nodes is called \emph{minimal} if it has the minimal cardinality among all  sets of driver nodes or sensor nodes, respectively.
The conditions in Theorem 1 allows finding a minimal set of driver nodes (resp. a minimal set of sensor nodes) by mapping the satisfaction of these conditions to solving maximum matching problem on the graph $\mathcal G$ (resp. $\mathcal G^\T$ obtained from $\mathcal G$ by reversing the direction of all its edges), see \cite{liu2016control}.

%
%


\section{Problem statement and Main results}
Here we generalize the analysis  of Section II by enlarging the class of dynamics that  the system can take to include arbitrary nonlinearities.
Specifically, we now consider general nonlinear systems of the form
\begin{equation}
\label{eq:nonlinear}
\dot x(t) = f\left(x (t), u(t) \right), \quad y(t) = h(x(t)),
\end{equation}
where  $f : \mathbb R^N \times \mathbb R^M \rightarrow \mathbb R^N$ and $h: \mathbb R^N \rightarrow \mathbb R^P$ are arbitrary \emph{meromorphic} functions of their arguments (i.e., each of their entries is the quotient of analytic functions).
The assumption of  meromorphic functions is very weak in the sense that it is satisfied by most models in biology, chemistry, ecology, and engineering (see, e.g., Table 1 in Ref. \cite{barzel2013universality}). 
Recall that meromorphic functions are either identically zero (written as ``$\equiv 0$'') or different from zero in an open dense subset of their domain (written as ``$\not\equiv 0$''), see  \cite[Chapter 1]{conte2007algebraic}.
This property allow us to define:

\begin{definition}  Given a meromorphic pair $\{f,h\}$, its graph   $\mathcal G_{f,h} = ({\sf X \cup \sf Y \cup \sf U}, {\sf A}_{f} \cup {\sf B}_{f} \cup {\sf C}_{h})$ has the edge-set defined as: $({\sf x}_j \rightarrow {\sf x}_i) \in {\sf A}_{f} \Leftrightarrow \partial f_i/ \partial x_j \not\equiv 0$; $({\sf u}_j \rightarrow {\sf x}_i) \in {\sf B}_{f} \Leftrightarrow \partial f_i / \partial u_j \not\equiv 0$; and $({\sf x}_j \rightarrow {\sf y}_i)  \in {\sf C}_{h} \Leftrightarrow \partial h_i / \partial x_j \not\equiv 0$.
\end{definition}



We say that two  pairs $\{f, h\}$ and $\{\tilde f, \tilde h\}$ are graph-equivalent if $\mathcal G_{f,h} = \mathcal G_{\tilde f,\tilde h} $. Since any $\{f,h\}$ is graph-equivalent to itself,  graph-equivalence is an equivalence relation.
Thus, given a graph $\mathcal G$, we can define the equivalence class
$$\mathfrak D(\mathcal G) := \{ \mbox{all meromorphic } \{f, h\}  \mbox{ such that }  \mathcal G_{f,h} = \mathcal G \}.$$
The class $\mathfrak D(\mathcal G)$ represents the set of all nonlinear dynamics that a system can have given that its graph is $\mathcal G$. 
Note that $\mathfrak D_L(\mathcal G)\subset \mathfrak D(\mathcal G)$.

As the nonlinear counterparts of linear controllability and linear observability, we consider the concepts of local \emph{accessibility} and local \emph{observability}.
We  will introduce  these concepts using the algebraic formalism of Ref. \cite{conte2007algebraic}.
Consider the field of meromorphic functions $\mathcal K$ in the variables $\{x, u, \dot u, \ddot u, \cdots, y, \dot y, \ddot y, \cdots\}$, and the  sets of differential symbols  $\ud x = (\ud x_1, \cdots, \ud x_N)^\T$, $\ud u = (\ud u_1, \cdots, \ud u_M)^\T$ and $\ud y^{(k)} = (\ud y_1^{(k)}, \cdots, \ud y_P^{(k)})^\T, k \geq 0$. 
For a function $\varphi \in \mathcal K$, its differential is $\ud \varphi = (\partial \varphi / \partial x)^\T \ud x + (\partial \varphi / \partial u)^\T \ud u$.
More generally, functions in the vector space  spanned over $\mathcal K$ by the elements of $\{\ud x, \ud u, \ud y, \cdots, \ud y^{(k)} \}$ are called \emph{one-forms}.
%
We next recall the following notions:

\begin{definition} \hfill
\begin{itemize}
\item[(i)] An \emph{autonomous element} of a system is a non-constant meromorphic function $\xi(x)$ such that its $k$-th time derivative $\xi^{(k)}$ is independent of $u$ for all $k \geq 0$, i.e., 
$$\partial \xi^{(k)} / \partial u \equiv 0,  \quad \forall k \geq 0.$$
\item[(ii)] A \emph{hidden element} of a system is a non-constant meromorphic function $\zeta(x)$  that is independent of $\{y, \cdots, y^{(k)}\}$ for all $k \geq 0$, i.e., 
$$ \ud \zeta \notin \mbox{span}_{\mathcal K} \{\ud y,  \ud \dot y, \cdots, \ud y^{(k)}  \}, \quad \forall k \geq 0. $$
\end{itemize}
\end{definition}

\begin{figure}[t!]
\begin{center}
\includegraphics[width=8cm]{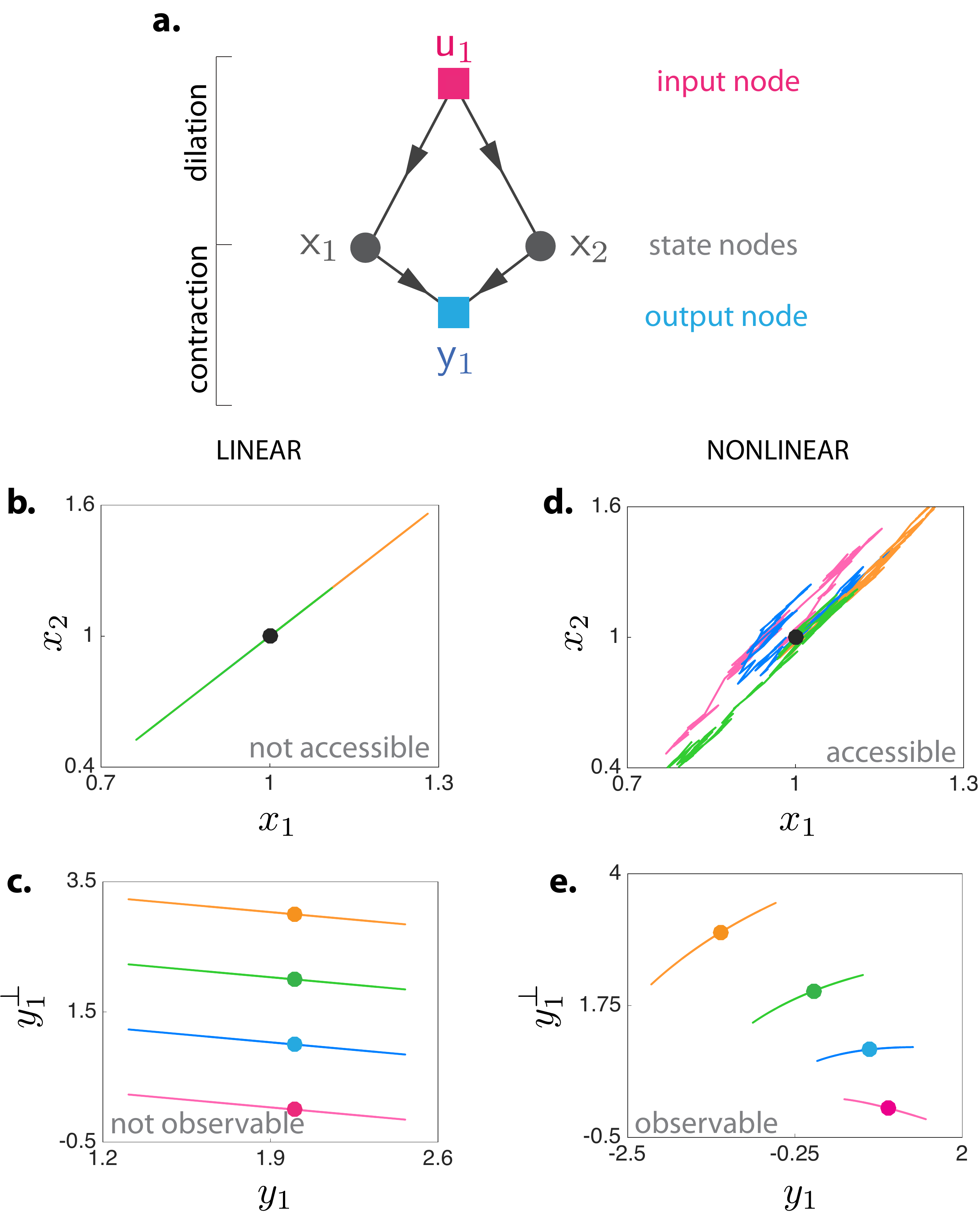}
\caption{{\bf a.} Graph of a system. Here the input, state and output nodes are ${\sf U} = \{{\sf u}_1\}$,  ${\sf X} = \{{\sf x}_1, {\sf x}_2,{\sf x}_3\}$ and ${\sf Y} = \{{\sf y}_1\}$, respectively.
%
{\bf b.} The graph of panel a is not linearly structurally controllable, meaning that no linear system with this graph is controllable. We illustrate this with five trajectories (colors) of the linear dynamics of Eq. \eqref{eq:linear-dilation} with initial condition $x(0) = (1,1)^\T$ (black dot), parameters  $b_{11} = 0.5$  and $b_{12}=1$, and random inputs $u_1(t)$. 
The lack of controllability constrains the system to the plane $\{x \in \mathbb R^2 | b_{21} x_1 - b_{11}x_2 = b_{21} x_1(0) - b_{11}x_2(0)\}$ for all time and inputs, representing the autonomous element of this system.
Consequently,  the system is not accessible and not controllable.
{\bf c.} The graph of panel a is not linearly structurally observable, meaning that no linear system with this graph is observable.
We illustrate this using the linear dynamics of Eq. \eqref{eq:linear-dilation}, where five different trajectories (colors) with different initial conditions (dots) give exactly the same projection in the output $y_1$ because they are all vertically aligned.
This is characterized by the hidden element $\zeta = c_{12} x_1 - c_{11} x_2$ which is orthogonal to the output and its derivatives.
Consequently, the system is not linearly observable and not locally observable.
{\bf d.} The nonlinear dynamics of Eq. \eqref{eq:nonlinear-dilation} is accessible because it lacks autonomous elements, illustrated here by five trajectories (colors) that are not constrained to any low-dimensional manifold.
Parameters are as in panel b and $\varepsilon = 0.5$.
This shows that the graph in panel a is structurally accessible.
{\bf e.} The nonlinear dynamics of Eq. \eqref{eq:nonlinear-dilation} is locally observable since different trajectories (colors) corresponding to different initial conditions (dots) give different projections in the output $y_1$.
Parameters are as in panel d and $\varepsilon=0.5$.
This shows that the graph in panel a is structurally locally observable.
}
\label{fig:unstable}
\vspace*{-0.5cm}
\end{center}
\end{figure}

An autonomous element constrains the state of the system to a low-dimensional manifold for all control inputs, just as in an uncontrollable linear system its state is constrained to a hyperplane.
%
%
A hidden element is an internal variable of the system  whose value cannot be inferred from the output, since it cannot be  rewritten as a function of the output and its derivatives.
A non-constant function that is not a hidden element is called observable.
%

With these notions a system is called locally \emph{accessible} (``accessible'', for short) if it does not have autonomous elements, and locally \emph{observable} (``observable'', for short) if it does not have hidden elements \cite{conte2007algebraic}.
%
%
%
%
For linear systems, the lack of autonomous elements is equivalent to linear controllability, and the lack of hidden elements is equivalent to linear observability  \cite{conte2007algebraic}.
For example, all linear systems of Eq. \eqref{eq:linear-dilation} are not controllable because  $\xi(x) = b_{11}x_1 - b_{21} x_2$ is an autonomous element for all of them.  Indeed  $ \xi^{(k)} \equiv 0$, which is independent of $u$ for all $k \geq 1$. 
Similarly,  $\zeta(x) = c_{12} x_1 - c_{11}x_{2}$ is a hidden element for all those linear systems, since $\zeta$ cannot be written as a function of $y_1 = c_{11}x_1 +c_{12}x_2$ and its derivatives $y_1^{(k)} = (c_{11} b_{11} + c_{12} b_{12})u_1^{(k-1)}$. Indeed, this happens because no output derivative contains  information of the state.
In this sense, the above definitions of accessibility and observability provide  nonlinear generalizations of linear controllability and linear observability.
%

In analogy to the definitions of linear structural controllability and observability, we now define:

\begin{definition} 
\label{def:struct-acc-obs}
$\mathfrak D(\mathcal G)$ is:
\begin{itemize}
\item[(i)] \emph{structurally accessible} if $\mathfrak D(\mathcal G)$ contains at least one system that is accessible.
\item[(ii)] \emph{structurally  observable} if  $\mathfrak D(\mathcal G)$ contains at least one system that is  observable.
\end{itemize}
\end{definition}

When $\mathfrak D(\mathcal G)$ is  structurally accessible (resp. structurally  observable), we also call the graph $\mathcal G$  structurally accessible (resp. structurally   observable).
We also call a particular $f$ structurally accessible if there exists at least one graph-equivalent $\bar f$ that is accessible. Similarly, a particular $\{f,h\}$  is structurally  observable if there exists at least one graph-equivalent $\{\bar f,\bar h\}$ that is  observable.

As in the  case of linear systems, in Lemma \ref{prop:prop1} of Section IV we prove that  in a structurally accessible class $\mathfrak D(\mathcal G)$ the subset of accessible systems is open and everywhere dense; furthermore the subset of non-accessible systems is not dense.
Similarly, in a structurally  observable class $\mathfrak D(\mathcal G)$, we prove in Lemma \ref{lemma:obs-generic} of Section V  that the subset of  observable systems is open and everywhere dense; in addition, the subset of non  observable systems is not dense.
This means that, if  $\mathfrak D(\mathcal G)$ is structurally accessible (resp. structurally  observable),  any of its systems  is either accessible (resp.  observable) or becomes accessible (resp.  observable) by an arbitrarily small change of its dynamics (see example Example 2 below).

Our main result is the following:

\begin{theorem}
\label{thm:main-theorem}
$\mathfrak D(\mathcal G)$ is:
\begin{itemize}
\item[(i)]  structurally accessible iff each state node  is the end-node of a path that starts in ${\sf U}$.
\item[(ii)] structurally  observable iff each state node  is the start-node of a path that ends in ${\sf Y}$.
\end{itemize}
\end{theorem}
\begin{proof} See Proposition \ref{prop:prop-struct-accessibility} in Section IV for point (i), and Proposition \ref{prop:struct-obs} in Section V for point (ii). \end{proof}

The above Theorem shows that despite the  observability of a nonlinear system may depend on which particular inputs are applied to it,  its structural  observability is independent of the inputs. 
This happens because removing all edges that connect the inputs to the state variables will not change if condition (ii) of Theorem 2 is satisfied.
Indeed, note  that including more edges in a graph cannot destroy its structural accessibility or structural  observability.
Note also that a ``duality'' similar to the case of linear systems remains: a network is structurally accessible if and only if its ``dual network'' (with reversed edges and the labels of input and output nodes interchanged) is structurally  observable.

In addition and somewhat counterintuitively, Theorem 2 shows that nonlinearities make it easier to ``control'' and ``observe'' networked systems because the conditions of  Theorem  2 are  weaker than those of Theorem 1. 
We illustrate this point by revisiting Example 1 now considering nonlinear dynamics:

\begin{example} For the graph in Fig. 1a, the class $\mathfrak D(\mathcal G)$   contains all  nonlinear systems of the form
\begin{equation}
\label{eq:nonlinear-dilation}
\left \{
\begin{split}
\dot x_1(t) &= b_{11} u_1(t) + \varepsilon, \\ 
\dot x_2(t) &= b_{21} u_1(t) + \varepsilon u_1^3(t),  \\
 y_1(t) &=  c_{11} x_1 + c_{12} x_2 + \varepsilon x_1(t) x_2(t),
\end{split}
\right.
\end{equation}
with nonzero $b_{11}, b_{21}$, and  $\varepsilon$. Note that Eq. \eqref{eq:nonlinear-dilation} is an ``$\varepsilon$-change'' of Eq. \eqref{eq:linear-dilation} because making $\varepsilon = 0$ renders Eq. \eqref{eq:nonlinear-dilation} equal to Eq. \eqref{eq:linear-dilation}. Note also:
\begin{itemize}
\item[1.] In the dilation of Fig. 1a, the nonlinearities in $\mathfrak D(\mathcal G)$ eliminate the autonomous element that was present in $\mathfrak D_L(\mathcal D)$. That is, the function $\xi (x)= b_{21} x_1  - b_{11} x_2 $ that was an autonomous element for all linear dynamics of Eq. \eqref{eq:linear-dilation} is no longer an autonomous element for Eq. \eqref{eq:nonlinear-dilation} because $\dot \xi = -   \varepsilon b_{11}  u_1^3$ depends on $u_1$.
This proves that $\mathfrak D(\mathcal G)$ is structurally accessible.
Indeed, the trajectories of Eq. \eqref{eq:nonlinear-dilation} are no longer constrained to a low-dimensional manifold (Fig. 1d).
\item[2.] In the contraction of Fig. 1a,  the nonlinearities in $\mathfrak D(\mathcal G)$ also eliminate the hidden element. To see this,  compute $\dot y_1 = \alpha_0(u_1) + \varepsilon \alpha_1 (u_1) x_1 + \varepsilon \alpha_2 (u_1) x_2,$ 
where $\alpha_0(u_1) = c_{11}[b_{11} u_1 + \varepsilon p_1] + c_{12} [b_{21}u_1 +  \varepsilon u_1^3]$, $\alpha_1(u_1) = b_{21}u_1 + \varepsilon u_1^3$ and $\alpha_2(u_1) = b_{11} u_1 + \varepsilon p_1$. 
Note  that $\alpha_1 \not\equiv 0$ and $\alpha_2 \not\equiv 0$ for almost all  $u_1$.
Therefore, the Jacobian
$$\begin{pmatrix}\partial y_1/ \partial x \\ \partial \dot y_1/ \partial x  \end{pmatrix}= \begin{pmatrix} c_{11} + \varepsilon x_2 & c_{12} + \varepsilon x_1 \\ \varepsilon \alpha_1 & \varepsilon \alpha_2 \end{pmatrix} $$ 
is generically nonsingular. Consequently, from the Implicit Function Theorem, it follows that   we can locally infer $x_1$ and $x_2$ from $y_1$ and $\dot y_1$. 
Indeed, the function $\zeta = c_{12} x_1 - c_{11} x_1$ that was a hidden element of the linear system of Eq. \eqref{eq:linear-dilation} is no longer a hidden element of Eq. \eqref{eq:nonlinear-dilation}. 
This proves that Eq. \eqref{eq:nonlinear-dilation} is  observable  (Fig. 1e), and that $\mathfrak D(\mathcal G)$ is structurally  observable.
%
%


\end{itemize}
\end{example}

\subsection{Minimal sets of driver/sensor nodes and input/output nodes.}

Consider graph $\mathcal G({\sf X}, {\sf A})$  consisting of state nodes $\sf X$ and edges $({\sf x}_i \rightarrow {\sf x}_j) \in  {\sf A}$.
We can extend  Definition \ref{def:driver-sensor-nodes} to nonlinear systems by requiring that a set of driver nodes ${\sf X}_D \subseteq \sf X$ renders $\mathcal G (\sf X \cup \sf U, \sf A \cup B)$ structurally accessible.
Similarly, a set of sensor nodes ${\sf X}_S \subseteq \sf X$ must render $\mathcal G (\sf X \cup \sf U, \sf A \cup C)$ structurally  observable.
  Then Theorem 2 has the following implication:

\begin{proposition}
 \hfill
 \label{prop:minimal-driver-sensor-nodes}
\begin{itemize}
\item[(i)] A minimal set of driver nodes is given by arbitrarily choosing one node in each root strongly-connected-component of $\mathcal G({\sf X},\sf A)$. 
\item[(ii)] A minimal set of sensor nodes is given by arbitrarily choosing one node in each top strongly-connected-component of $\mathcal G({\sf X}, \sf A)$.
\end{itemize}
\end{proposition}

A strongly connected component (SCC) of a graph  $\mathcal G$ is a maximal subgraph such that there is a directed path in both directions between any two of its nodes \cite[pp. 552-557]{cormen2009introduction}.
  A root SCC is an SCC without incoming edges, and a top SCC is an SCC without outgoing edges.
Recall that any directed graph can be decomposed into an acyclic graph between its SCCs, with root and top SCCs at the start and end of this graph, respectively \cite{cormen2009introduction}. 
 Let $m$ be the number of root SCCs and $p$ the number of top SCCs of $\mathcal G({\sf X}, \sf A)$. 
 Then a proof of Proposition 1-(i) is obtained from the fact that if a single input node $\sf u$ is connected to one arbitrary node ${\sf x}_j$ of each root SCC (i.e., ${\sf u} \rightarrow {\sf x}_j$, $j=1,\cdots, m$), the decomposition into the acyclic graph of SCC implies that the graph  satisfies condition (i) of Theorem \ref{thm:main-theorem}.
 Analogously, a proof of Proposition 1-(ii) is obtained from the fact that if  a single output node $\sf y$ is connected with one arbitrary node ${\sf x}_j$ of each top SCC  (i.e., $ {\sf x}_j \rightarrow {\sf y}$, $j=1,\cdots, p$), this will yield a graph that satisfies condition (ii) of Theorem \ref{thm:main-theorem}.
An additional consequence of this argument is the following:

 \begin{corollary}  \hfill
\begin{itemize}
\item[(i)]  The minimal number of driver nodes of any graph is its number of root SCCs, and the minimal number of sensor nodes is its number of  top SCCs.
\item[(ii)] The minimal number of input nodes that renders any graph structurally accessible is always one, and the minimal number of output nodes that renders any graph  structurally observable is also one.
 \end{itemize}
 \end{corollary}
 
The second statement in the above Corollary generalizes  the result of Ref. \cite{kawano2016any} to structural systems and to the case of analyzing observability.
 
%
%
All minimal sets of driver or sensor nodes of arbitrary graphs can be found in linear time, since the SCCs of general graphs can be computed in linear time  \cite[pp. 35]{cormen2009introduction}.
For comparison, in the case of linear structural accessibility (resp. linear structural observability), solving the maximum-matching problem to find one set of  driver nodes (resp. sensor nodes) takes polynomial time, and identifying all sets of driver nodes (resp. sensor nodes) is intractable for large graphs.

 
 
 The following two Sections build the proofs for our main results.



\section{Proof of  the Structural Accessibility Theorem}


Given a graph $\mathcal G = ({\sf X \cup \sf U} , \sf A \cup B)$, here we consider the class $\mathfrak D(\mathcal G)$ of all controlled systems
\begin{equation}
\label{eq:nonlinear-control}
 \dot x(t) = f(x(t), u(t)),
\end{equation}
such that $\mathcal G_{f} = \mathcal G$. 
%
%

Our first result shows that accessible systems are ``generic'' in a structurally accessible class $\mathfrak D(\mathcal G)$, while non-accessible are not (i.e., they are ``hard to find'').
To establish this result, our argument relies on the notions of the  \emph{$k$-jet} $f_k$ of the meromorphic function $f$ ---informally defined as taking the first $k$-terms of its Taylor expansion--- and the resulting topology that can be constructed ---the so-called  ``Whitney $C^k$ topologies''.
We refer the reader to \cite[Section 2.1]{arnold1972lectures} and \cite[Chapter II.3]{golubitsky2012stable} for further details.
Specifically,  the topology we  use is defined from the notion of an open ball of radius $\varepsilon \geq 0$ centered at  a meromorphic $f_0$. This ball consists of all meromorphic $f$'s for which $\exists k \in \mathbb N$ such that the Euclidean distance between the first $\ell$ Taylor coefficients of $f_0$ and $f$ is less than $\varepsilon$ for all $\ell \geq k$.

%

%

\begin{figure}[t!]
\begin{center}
\includegraphics[width=7cm]{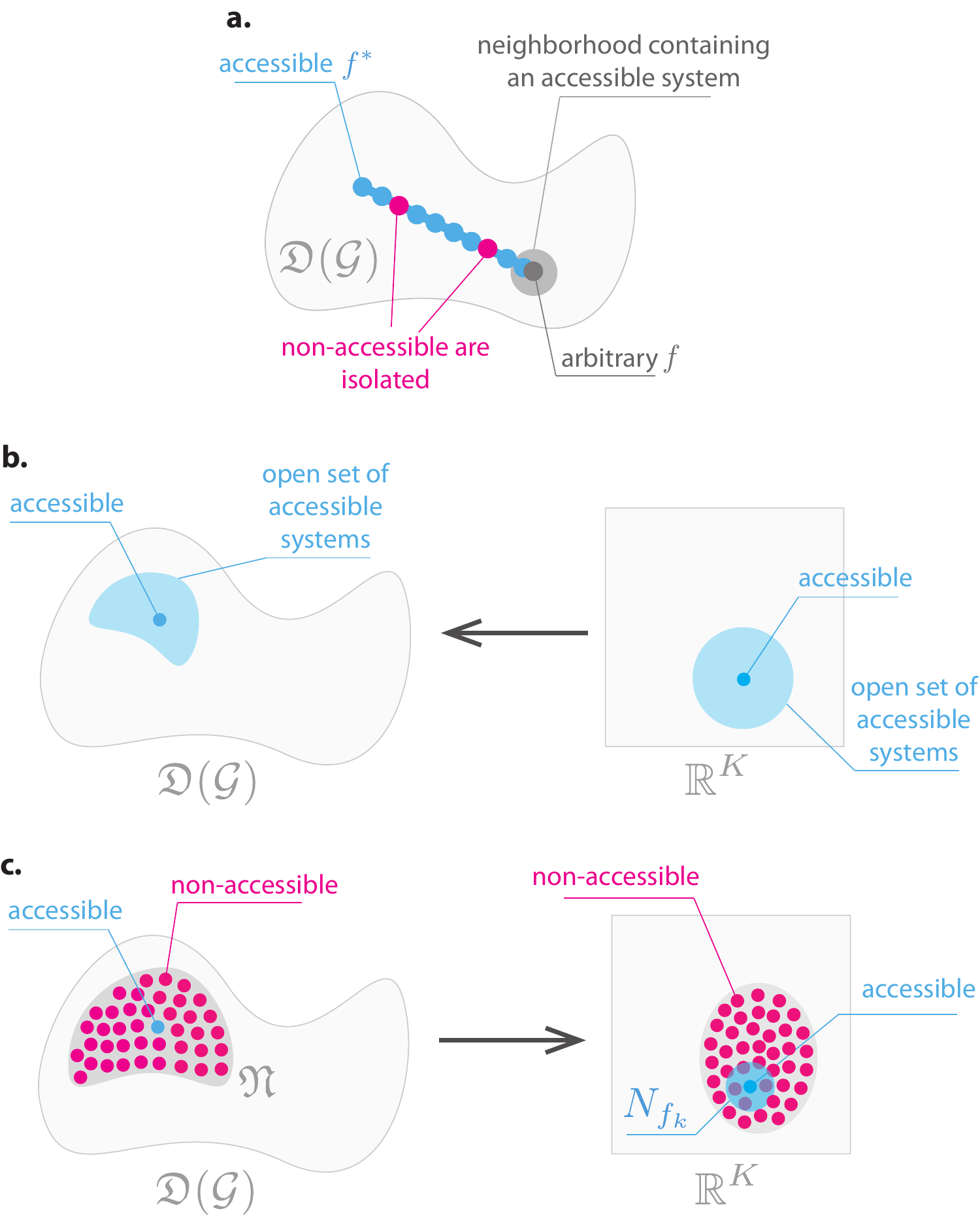}
\caption{ Figures for the proof of Lemma 1. {\bf a.} Accessible systems are dense. {\bf b.} The subset of accessible systems is open. {\bf c.} Contradiction obtained by assuming that  the subset of non-accessible systems is dense, proving that non-accessible systems are not dense.
%
%
}
\label{fig:Lemma1}
\vspace*{-0.5cm}
\end{center}
\end{figure}



\begin{lemma}
\label{prop:prop1}
 If $\mathfrak D(\mathcal G)$ is structurally accessible then: (i) the subset of accessible systems is dense everywhere in $\mathfrak D(\mathcal G)$; (ii) the subset of accessible systems is open; and (iii) the subset of non-accessible systems is not dense.
\end{lemma}
\begin{proof} Although (iii) $\Rightarrow$ (i) because $\mathfrak D(\mathcal G)$ is the disjoint union of accessible and non-accessible systems, independent proofs of each statement are provided:
\hfill
\begin{itemize}
\item[(i)] We show that any $f \in \mathfrak D(\mathcal G)$ has an arbitrarily close neighbor $\tilde f \in \mathfrak D(\mathcal G)$ that is accessible (Fig. 2a).
Let $f^* \in \mathfrak D(\mathcal G)$ be an accessible system in $\mathfrak D(\mathcal G)$ (there is at least one because of the definition of structural accessibility).
Define the convex combination $ f_\lambda = \lambda f^* + (1-\lambda) f$.
Note that  $ f_\lambda \in \mathfrak D(\mathcal G)$ for almost all $\lambda \in [0,1]$.
Note also that for $\lambda = 1$ we have $f_1 = f^*$, implying that  $f_1$ is accessible. 
Consequently, due to the generic properties of meromorphic  functions and the Accessibility rank condition \cite{conte2007algebraic}, the family of systems $\{ f_\lambda\}$ are accessible for almost all $\lambda \in [0,1]$. Therefore, for any $\varepsilon >0$, there exists $\lambda^* >0$ such that $\lambda^*< \varepsilon$ and $ f_{\lambda^*}$ is accessible.
Thus, $f_{\lambda^*}$ is an $\varepsilon$-neighbor of $f$ which is accessible, completing the proof.
\item[(ii)] 
We prove  that any accessible $f \in \mathfrak D(\mathcal G)$ has a neighborhood consisting only of accessible systems.
Since $f$ is meromorphic, we can rewrite this function as the Taylor-expansion $f(x,u) = \alpha_0(x) + \sum_{i=1}^\infty \alpha_i(x) u^i$ with $\alpha_i \in \mathcal K$. Note that the accessibility of $f$ implies there exists a $k \in \mathbb \mathbb N$ such that the $k$-jet $f_k(x,u) :=\alpha_0(x) + \sum_{i=1}^k \alpha_i(x) u^i$ is accessible.
Indeed, since $f$ is accessible there cannot be autonomous elements $\xi \in \mathcal K$, implying that $\ud \xi$ is not orthogonal to at least some $\alpha_k$, $k \in \mathbb N$. 
 This implies that no (non-constant) $\xi \in \mathcal K$ can be an autonomous element for $f_k$, making the $k$-jet $f_k$ accessible.
 
%
%
%
Recall that this $k$-jet represents the first $k$ terms of the Taylor expansion of $f$, implying we can associate $f_k$ to a point in $\mathbb R^K$ for some $K$ that depends on $k$ (right in Fig. 2b).
Next we regard $f_k$ as a polynomial function of its Taylor coefficients, so that the generic properties of meromorphic functions imply that $f_k$ has a neighborhood $N_{f_k} \subseteq \mathbb R^K$ of accessible systems.
All $\tilde f \in \mathfrak D(\mathcal G)$ such that their $k$-jets $\tilde f_k$ satisfy $\tilde f_k \in N_{f_k}$ will form the open neighborhood of $f$ of accessible systems.

\item[(iii)] We prove by contradiction, assuming that $\mathfrak D(\mathcal G)$ is structurally accessible but that it contains an open set $\mathfrak N$ such that non-accessible systems are dense on it (pink in Fig. 2c).
Since $\mathfrak D(\mathcal G)$ is structurally accessible and accessible systems are dense due to Lemma \ref{prop:prop1}-(i), then $\mathfrak N$ contains at least one accessible system $f$ (blue in Fig. 2c).
Now choose $k \geq 0$ large enough such that the $k$-jet $f_k$ of the accessible system $f$ is accessible. 
The $k$-jets $\tilde f_k$ of all non-accessible systems $\tilde f$'s remain non-accessible.
Since the $f_k$ and the $\tilde f_k$'s represent the first $k$ terms of the Taylor expansion of $f$ and the $\tilde f$'s,  we can associate each of them to a point in $\mathbb R^K$ corresponding to the value of the first $k$ coefficients of their  Taylor expansion  (here again $K$ is some constant that depends on $k$). 
Since $\mathfrak N$ is a neighborhood of $f$, all its elements are mapped to a corresponding neighborhood of $f_k$ in $\mathbb R^K$ such that the points corresponding to non-accessible systems are dense (Fig. 2c).
Considering now that $f_k$ is accessible and that it is a polynomial function of its Taylor coefficients, the generic properties of meromorphic functions imply that there exists a neighborhood $N_{f_k} \subseteq \mathbb R^K$ of $f_k$  such that all its corresponding elements are accessible (blue neighborhood in Fig. \ref{fig:Lemma1}c).
This gives the desired contradiction, since it contradicts the fact that the non-accessible systems were dense.
\end{itemize}
\end{proof}

The next result allows us to analyze the structural accessibility of a graph from its \emph{spanning} subgraphs, which  will be instrumental for the proof of the main result. Recall that a subgraph $\tilde {\mathcal G}$ of $\mathcal G$  is spanning when $\tilde {\mathcal G}$ includes all nodes of $\mathcal G$.
\begin{lemma}
\label{prop:prop2} Let $\tilde {\mathcal G} \subseteq \mathcal G$ be any spanning subgraph of $\mathcal G$. If $\mathfrak D(\tilde {\mathcal G})$ is structurally accessible then $\mathfrak D( {\mathcal G})$ is also structurally accessible.
\end{lemma}
\begin{proof}  Since $\mathfrak D(\tilde {\mathcal G})$ is structurally accessible, it contains one system $\dot x = f(x,u)$ which is accessible. Notice that starting from $\tilde{ \mathcal G}$, we can recover $\mathcal G$ by adding some edges. Suppose that the edge ${\sf x}_j \rightarrow {\sf x}_i$ is added to $\tilde{ \mathcal G}$ to obtain $\mathcal G$. Then $\mathfrak D(\mathcal G)$  contains the systems
\begin{equation}
\label{eq:f+alpha}
\dot x_i = f_i(x, u) + \alpha  x_j, 
\end{equation}
for any constant $\alpha \neq 0$. Similarly, if the edge ${\sf u}_j \rightarrow {\sf x}_i$ is added $\mathfrak D(\mathcal G)$  contains the systems
\begin{equation}
\label{eq:f+alpha1}
\dot x_i = f_i(x, u) + \alpha  u_j.
\end{equation}
 For $\alpha = 0$ the systems of Eqs. \eqref{eq:f+alpha} or \eqref{eq:f+alpha1} are accessible. Additionally, their right-hand side is a meromorphic function of $\alpha$. Thus, due to the generic properties of meromorphic functions  \cite{conte2007algebraic},  both systems are accessible for almost all $\alpha \in \mathbb R$. Therefore, the class $\mathfrak D(\mathcal G)$ is structurally accessible. Repeating the same argument for all other edges completes the proof.
\end{proof}


Now consider a meromorphic function $\varphi(x, u): \mathbb R^N \times \mathbb R^M \rightarrow \mathbb R^N$ and a subset of nodes ${\sf V} \subseteq {\sf X \cup \sf U}$. 
We write $\varphi \in {\sf S}$ if $\varphi(x)$  depends on \emph{all}  variables $v_i$ for all ${\sf v}_i \in {\sf V}$.
%
%
%
With this notation, an autonomous element of Eq. \eqref{eq:nonlinear-control}  is a non-constant meromorphic function $\varphi(x)$ such that  $\varphi^{(k)} \not\in {\sf U}$ for all $k \geq 0$.

\begin{example} For the graph of Fig. 1a with the linear dynamics of Eq. \eqref{eq:linear-dilation} we have that $\xi = b_{11} x_2 - b_{21} x_1$ satisfies $\xi^{(k)} = 0$ for  all $k \geq 1$. Thus we have that $ \xi^{(k)} \not\in \sf U$ for all $k$ and hence $\xi$ is an autonomous element.
\end{example}

Next, for a set $\sf S$ of state nodes, define its ``tail-set''  $T(\sf S) \subseteq {\sf V} \cup {\sf U}$ as  all nodes which point to $\sf S$ (Fig. 3a).
We denote    $T^k( {\sf S}) :=  T( T^{k-1} (\sf S))$.
%

\begin{figure}[t!]
\begin{center}
\includegraphics[width=8.5cm]{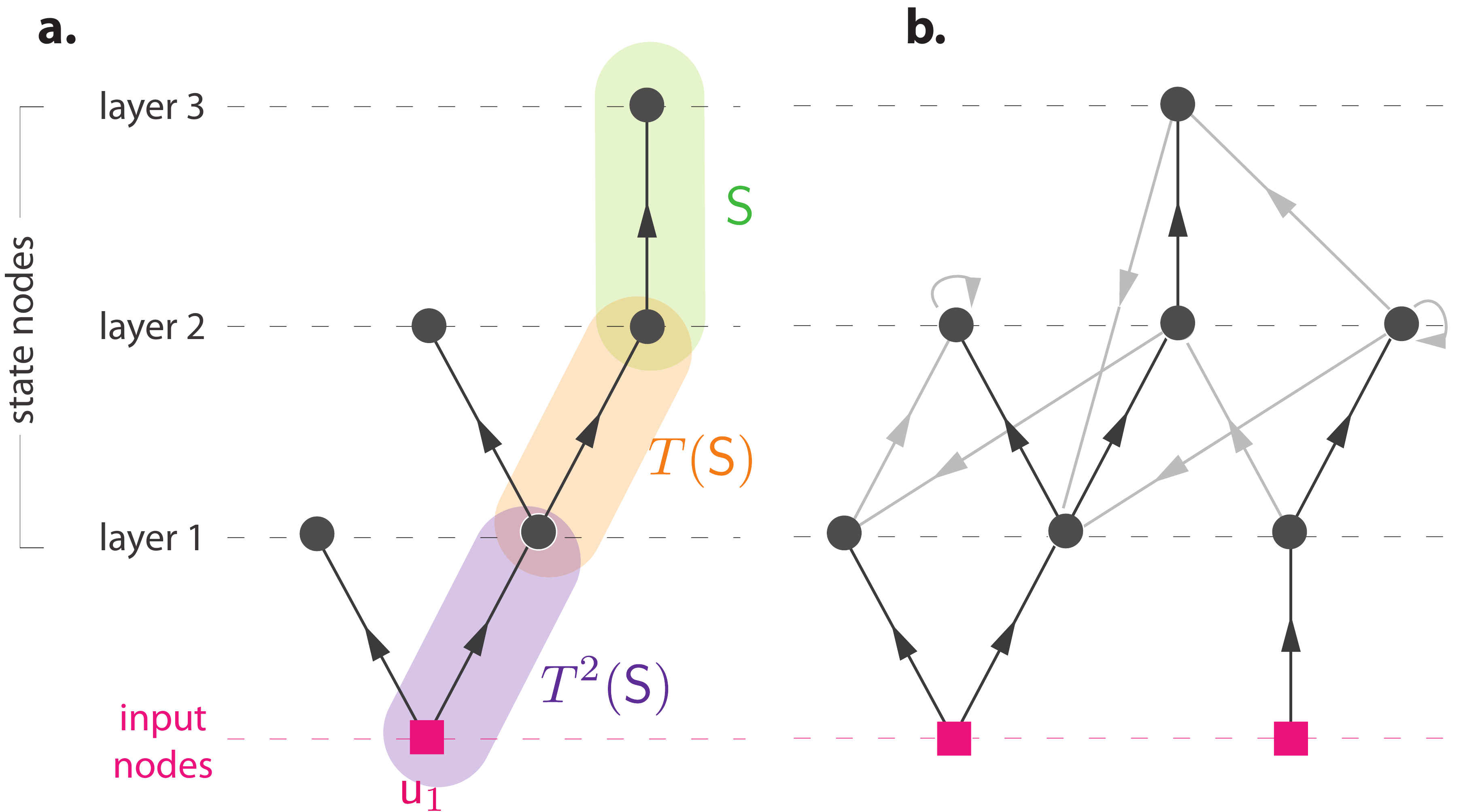}
\caption{
{\bf a.} A tree graph $\mathcal G$ where each state node has only one incoming edge and its root is an input node. 
A set $\sf S$ and its tail-sets $T(\sf S)$ and $T^2(\sf S)$ are marked in green, orange and purple, respectively. 
{\bf b.} From the graph $\mathcal G$ with  $M=2$ input nodes (dark and light edges),  a spanning subgraph $\tilde {\mathcal G}$ (dark edges) is obtained that has  $M$ disjoint trees,  one incoming edge per state node, and  roots at the input nodes. 
}
\label{fig:unstable}
\vspace*{-0.5cm}
\end{center}
\end{figure}

\begin{example} 
\label{ex:example-tree}
Consider a graph $\mathcal G$ which is a (connected) directed tree with each state node ${\sf x}_i$ having a single incoming edge, and rooted at a single input node ${\sf u}_1$ (Fig. 3a).
Note that the state nodes can be organized into $L$ layers according to the distance they have to the input node, with the first layer being all state nodes with distance one.
Consider the polynomial dynamics
\begin{equation}
\label{eq:polynomial-eq}
\dot x_i = (f_{T(i)})^{p_i} 
\end{equation}
where $f_{T(i)} = x_{T(i)}$ if ${\sf x}_i$ is in layer $\geq 2$, and $f_{T(i)} = u_1$ otherwise.
The vector $p = (p_1, \cdots, p_N) \in \mathbb N_+^N$ contains $N$ different integers with $\min_k p_k$ large enough.

For this graph $\mathcal G$ and the dynamics of Eq. \eqref{eq:polynomial-eq},   any non-constant meromorphic function $\varphi \in {\sf S}$ satisfies $\dot \varphi \in T({\sf S})$ for any $\sf S \subseteq \sf X$.
Namely,  if $\varphi$ depends on $\{x_i, \cdots, x_k\}$, then $\dot \varphi$ depends on all variables $\{f_{T(i)}, \cdots, f_{T(k)}\}$.
To show this, just note that
$$\dot \varphi = \sum_{i \in \sf S} \frac{\partial \varphi}{ \partial x_i} \dot x_i =  \sum_{i \in \sf S} \frac{\partial \varphi}{ \partial x_i} f_{T(i)}^{p_i} \in  T({\sf S}),$$
and that no term can cancel out in the sums because they have different exponents.

This observation allow us to prove that this system is accessible.
Indeed, take any $\sf S \subseteq \sf X$ and any non-constant meromorphic function $\varphi \in \sf S$. Since all state nodes are the end-node of a $\sf U$-rooted path, there exists a finite $k$ such that ${\sf u}_1 \in T^k (\sf S)$. Since  $\varphi^{(k)} \in T^k ({\sf S})$, this implies that $\varphi$ cannot be an autonomous element.
\end{example}

Combining Example \ref{ex:example-tree} with Lemma \ref{prop:prop2}, we have actually proved the following result:
\begin{lemma}
\label{lem:lemma-1} Assume that $\mathcal G$ is spanned by a disjoint union of directed trees rooted at $\sf U$, with each state node having a single incoming edge. Then $\mathfrak D(\mathcal G)$ is structurally accessible.
\end{lemma}

We now have all the ingredients for proving our main result:
\begin{proposition}
\label{prop:prop-struct-accessibility}
 $\mathfrak D(\mathcal G)$ is structurally accessible iff each state node  is the end-node of a path that starts in ${\sf U}$.
\end{proposition}
\begin{proof} \hfill
\begin{itemize}
\item[($\Leftarrow$)] By contradiction. If there is a state node ${\sf x}_i$ that is not the end-node of any ${\sf U}$-rooted path, then $x_i$ itself is an autonomous element.
\item[($\Rightarrow$)] Since each state node is the end-node of a $\sf U$-rooted path, note we can always obtain a spanning subgraph $\tilde {\mathcal G}$ of $\mathcal G$ such that: (i) it is a disjoint union of (connected) directed trees rooted at $\sf U$; (ii) each state node has a single incoming edge (Fig. 2b).
 %
 By Lemma \ref{lem:lemma-1}, the class $\mathfrak D( {\mathcal G})$ is structurally accessible. 
 
\end{itemize}
\end{proof} 

\begin{remark}
\label{rem:trivial-accessibility}
Note that in the trivial cases of an empty graph (i.e., a graph without nodes) or a graph without state nodes (i.e.,  the underlying system has no dynamics), applying Definition \ref{def:struct-acc-obs} yields that both graphs are structurally accessible because the set of autonomous element is empty.
\end{remark}

\begin{remark} Note that, even if $\mathcal G$ is linearly structurally controllable, this does not imply that all nonlinear systems with graph $\mathcal G$ are accessible.
For example, the graph corresponding to the system  $ \dot x_1 = x_2 + x_3 u$, $\dot x_2 = -x_1$, and $  \dot x_3 = -x_1 u$
is linearly structurally controllable. Yet, this  nonlinear system is not accessible because $\xi = x_1^2 + x_2^2 + x_3^2$ is an autonomous element.
\end{remark}

\begin{remark} Note that restricting the system dynamics of Eq. \eqref{eq:nonlinear-control} to be affine in the control input  changes the graph conditions for structural accessibility. In such case, graphs that contains ``pure dilations of the control input'' as  in Fig.1a are not structurally accesible because those subgraphs only admit  linear dynamics
\end{remark}


\section{Proof of  the Structural  Observability Theorem}

We start with the following observation:

\begin{lemma}
\label{lemma:obs-generic}
\hfill
\begin{itemize}
\item[(i)] If $\mathfrak D(\mathcal G)$ is structurally  observable then the subset of  observable systems is open and dense everywhere in $\mathfrak D(\mathcal G)$; furthermore, the subset of non observable systems is not dense.
\item[(ii)] Let $\tilde{\mathcal G} \subseteq \mathcal G$ be any spanning subgraph of $\mathcal G$. If $\mathfrak D(\tilde{\mathcal G})$ is structurally  observable then $\mathfrak D(\mathcal G)$ is also structurally  observable. 
\end{itemize}
\end{lemma}
\begin{proof} A proof for item (i) follows using the exact same argument as in the proof of Lemma \ref{prop:prop1}. Similarly, item (ii) follows using the same argument as in the proof of Lemma \ref{prop:prop2}.
\end{proof}

We next prove the structural  observability of a particular class of graphs:

\begin{lemma}
\label{lemma:observability-tree}
 Suppose that $\mathcal G(\sf X \cup Y, A \cup C)$ is a (connected) directed tree topped at a single output node $\sf y$, with each state node having a single outgoing edge. Then $\mathfrak D(\mathcal G)$ is structurally  observable.
\end{lemma}
\begin{proof} From the structure of the graph we can order its nodes by layers, where nodes with distance $k$ to the output $\sf y$ belong to the $k$-th layer  (Fig. 4a).
We will prove the claim by induction in the number of layers:
\begin{itemize}
\item[(i)] For one layer, denote its nodes by $\{ {\sf x}_{1}, \cdots, {\sf x}_{d_1} \}$ where $d_1$ is the number of nodes.  One particular dynamics admissible for this graph is
\begin{equation}
\label{eq:sys-one-layer}
\dot x_{i} = c_i, \ i =1, \dots, d_1;  \quad y = x_{1} x_{2} \cdots x_{d_1},
\end{equation}
with $c_i$ some non-zero constants.
In the following we show that  Eq. \eqref{eq:sys-one-layer} is  observable by proving that the span of $\ud y$ and its derivatives $\ud y^{(k)}$ equals  $ \span_{\mathcal K} \ud x$.
If $d_1 = 1$ the claim follows directly, because there is only one state variable $x_1$ and $y = x_1$ renders it observable.
Consider now that $d_1 >1$.
From direct calculation we obtain the identity:
$$\dot y = \frac{\ud}{\ud t} \left( \prod_{i=1}^{d_1} x_i\right) =  \left( \prod_{i=1}^{d_1} x_i\right)  \left( \sum_{i=1}^{d_1} \frac{c_i}{x_i}\right)  = ( y)  \left( \sum_{i=1}^{d_1} \frac{c_i}{x_i}\right).$$
The variable $z := \dot y / y$  is  observable from $y$.
Therefore, the system of Eq. \eqref{eq:sys-one-layer} will be  observable if the span of $\ud z$ and its derivatives $\ud z^{(k)}$ equals  $\span_{\mathcal K} \ud x$.
Note that
$$z^{(k)} = \sum_{i=1}^{d_1} (-1)^k k! \ \frac{c_i^{k+1}}{x_i^{k+1}},  $$
so its differential is
$$\ud z^{(k)} = \sum_{i=1}^{d_1}(-1)^{k+1} (k+1)!  \ \frac{ c_i^{k+1}}{x_i^{k+2}} \ud x_i. $$
Taking $k = d_1+1$, the set $\{\ud z, \cdots, \ud z^{(d_1+1)} \}$ will contain the functions $\{1/x_i^2, \cdots, 1/x_i^{d_i + 3} \}$, $i = 1,\cdots, d_1$,  whose span is $\span_{\mathcal K} \ud x$. This proves that the system of Eq. \eqref{eq:sys-one-layer} is  observable, and thus that a graph $\mathcal G$ with one layer is  locally observable.
\item[(ii)] For the induction step, we show that  if a graph $\mathcal G$ with $L$ layers is structurally observable, then a graph with $L+1$ layers is also structurally observable.
By definition,  the nodes in the $(L+1)$-th layer are only connected to nodes in the $L$-th layer. Furthermore, they are connected in the same way as nodes in the first layer are connected to the output node (Fig. 4a).
Therefore, the  argument  in point (i) with $y$ replaced by the corresponding node in the $L$-th layer implies that the nodes in the $(L+1)$-th layer are observable. 
This completes the proof. 
\end{itemize}
\end{proof}

\begin{figure}[t!]
\begin{center}
\includegraphics[width=9cm]{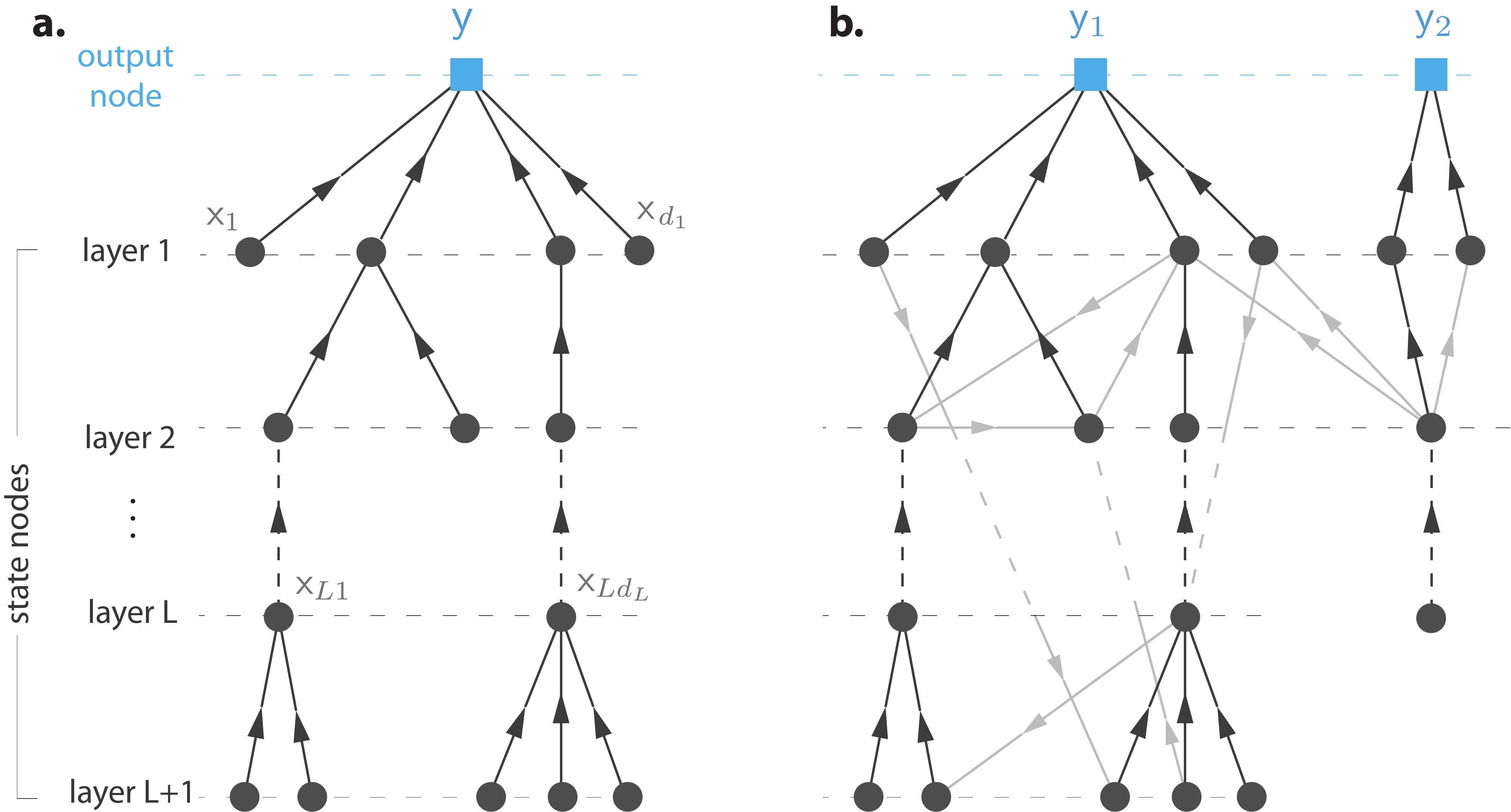}
\caption{
{\bf a.} A tree graph $\mathcal G$ topped at ${\sf y}$ with a single outgoing edge per  state node. 
{\bf b.} From any graph $\mathcal G$ such that each node has a path to  $\sf y$  (dark and light edges),  a subgraph $\tilde {\mathcal G}$ (dark edges) can be obtained such that it is a tree topped at $\sf y$ and each state node has a single outgoing edge. 
}
\label{fig:unstable}
\vspace*{-0.5cm}
\end{center}
\end{figure}

The final result follows by decomposing the graph into disjoint trees topped at the output nodes:


\begin{proposition}
\label{prop:struct-obs}
 $\mathfrak D(\mathcal G)$ is structurally  observable iff each state node  is the start-node of a path that ends in ${\sf Y}$.
\end{proposition} 
\begin{proof} \hfill
\begin{itemize}
\item[($\Leftarrow$)] By contradiction. If there is a state node ${\sf x}_i$ that is not the start-node of any ${\sf Y}$-topped path, then $x_i$ itself is a hidden element.
\item[($\Rightarrow$)] Since each state node is the start-node of a $\sf Y$-topped path, note we can always obtain a spanning subgraph $\tilde {\mathcal G}$ of $\mathcal G$ such that: (i) it is a disjoint union of (connected) directed trees topped at $\sf Y$; (ii) each state node has a single outgoing edge (Fig. 4b).
 By Lemma \ref{lemma:observability-tree}, $\mathfrak D( \tilde {\mathcal  G})$ is structurally  observable. Since $\tilde {\mathcal G} \subseteq \mathcal G$ is a spanning subgraph, Lemma \ref{lemma:obs-generic}-(ii) implies that $\mathfrak D (\mathcal G)$ is structurally  observable. 
\end{itemize}
\end{proof}

\begin{remark}
\label{rem:trivial-accessibility}
In analogy to Remark 1, in the trivial cases of an empty graph (i.e., a graph without nodes) or a graph without state nodes (i.e.,  the underlying system has no dynamics), applying Definition \ref{def:struct-acc-obs} yields that both graphs are structurally  observable because the set of hidden elements is empty.
\end{remark}

%

%


\section{Discussion and Concluding Remarks}

The notions of structural accessibility and structural observability that we have introduced and characterized  are nonlinear counterparts of the notions of linear structural controllability and linear structural observability. 

We next discuss some testable predictions offered by our theory. In a recent study of the locomotion of the worm \emph{C. elegans}, the ablation of the neuron PDB was found to generate a dilation in the nervous system connectome that decreased its (output) structural linear controllability \cite{yan2017network}. This loss of linear controllability was suggested to imply that the worm lost some ``directions'' in which it is was able to move, which were experimentally confirmed by a decreased ability to produce some specific motion patterns (quantified by a decrease in certain so-called ``eigenworms''). 
Assuming that the nervous system of the \emph{C. elegans} is an arbitrary nonlinear system instead of a linear system, our theory implies that the dilation caused by ablating PBD cannot decrease the structural accessibility of the \emph{C. elegans} connectome. Namely, the nervous system of an ablated worm can reach the same set of states as those of normal worms using perhaps different ``longer'' trajectories (e.g., by using different paths in the connectome that yield different combinations of ``eigenworms'').  Thus,  our structural accessibility theory predicts that PDB ablated worms can still adopt each body pose that a non-ablated worm can adopt. 
More generally, we predict that the ability of a worm to adopt a body pose is preserved as long as the ablated interneurons do not fully disconnect an input (i.e., a sensory neuron) or an output (i.e., a motor neuron).

We emphasize that more detailed predictions for the impact of the network structure on the controllability or observability properties can be obtained when the class of dynamics that the system can take is better known ---such as  neuronal,  ecological, gene regulatory, or epidemic systems, see e.g., \cite{barzel2013universality}. 
Such an analysis would provide graph conditions for structural accessibility and structural local observability that are ``between'' those of Theorem 1 (i.e., linear systems), and those of Theorem 2 (i.e., arbitrary nonlinear systems).  
Indeed, note that the conditions of Theorem 2 are always necessary, but they may not be sufficient when we restrict the system dynamics to belong to a particular class.
For example, in \cite{angulo2017controlling} and \cite{aparicio2018}, we analyzed the structural accessibility and structural local observability properties for the particular class of nonlinear dynamics found in ecosystems.
In this analysis, we found  that the conditions for structural accessibility and structural local observability for  ecological dynamics are indeed stronger than those of Theorem 2.

Finally, our results provide a broader perspective of what  we  can deduce about the controllability or observability properties of a system from knowing only its interconnection network. We have shown that if the control inputs can reach all state nodes through a path in the network,  then there exists some admissible system dynamics that is accessible. Similarly, if all state nodes can reach an output through a path in the network, then there exists some admissible system dynamics that is locally observable. These two facts suggest that the interconnection network only encodes the essential information of the controllability and observability properties of complex systems.

\vspace*{0.2cm}

{{\it Acknowledgements.}  
M.T.A. gratefully acknowledges the financial support from CONACyT, M\'exico, and LS2N, France.
We also thank Yang-Yu Liu and Yu~Kawano for  insightful comments about preliminary versions of this paper.
}

\ifCLASSOPTIONcaptionsoff
  \newpage
\fi


\bibliographystyle{IEEEtran}
\bibliography{IEEEabrv,ContractionTheory}
%

\end{document}